\documentclass[12pt,fleqn,reqno]{amsart} 

\usepackage{amsrefs}  
\usepackage{amsthm}
\usepackage{epsfig}
\usepackage{graphicx}
\usepackage{enumerate}
\usepackage{color} 
\usepackage[toc,page]{appendix}
\usepackage[dvipsnames]{xcolor}
\usepackage{bm}
\usepackage{dsfont} 
\newcommand{\mathbbm}[1]{\mathds{#1}}
\usepackage{cancel}
\usepackage{enumitem}
\usepackage{physics}
\usepackage{amsmath,euscript}
\usepackage{amsthm}
\usepackage{amssymb}
\usepackage{graphicx}
\usepackage{mathrsfs}
\usepackage{epsf}
\usepackage{xcolor}
\usepackage{psfrag}
\usepackage{enumerate}
\usepackage{amsfonts,amssymb,amsthm,amsmath} 
\usepackage{hyperref}

\usepackage{enumitem} 
\usepackage{xspace}
\usepackage[margin=1.37in]{geometry}
\usepackage{cancel}

\usepackage{todonotes} 
\usepackage{float} 
\usetikzlibrary{decorations.markings}

\usepackage{soul}

\newtheorem{theorem}{Theorem}[section]
\newtheorem{corollary}[theorem]{Corollary}

\newtheorem{lemma}[theorem]{Lemma}
\newtheorem{prop}[theorem]{Proposition}

\theoremstyle{remark}
\newtheorem{remark}[theorem]{Remark}

\theoremstyle{definition}

\numberwithin{thm}{section}
\numberwithin{equation}{section}

\usepackage{soul}
\setstcolor{red}

\definecolor{green}{rgb}{0.0, 0.5, 0.5}
\definecolor{yellow}{rgb}{0.5, 0.5, 0}
\definecolor{lgray}{gray}{0.9}
\definecolor{llgray}{gray}{0.95}
\definecolor{lllgray}{gray}{0.975}
\newcommand{\red}{\color{black}}

\newcommand{\nc}{\newcommand}

\nc{\la}{\label}
\nc{\ba}{\begin{array}}
\nc{\ea}{\end{array}}
\nc{\bs}{\begin{split}}
\nc{\es}{\end{split}}

\newcommand{\R}{\mathbb{R}}

\nc{\al}{\alpha}
\nc{\G}{\Gamma}
\nc{\et}{\eta} 
\nc{\g}{\gamma}
\nc{\gam}{\gamma}
\nc{\ka}{\kappa}
\nc{\lam}{\lambda}
\nc{\Lam}{\Lambda}
\nc{\Om}{\Omega}
\nc{\om}{\omega}
\nc{\ta}{\tau}
\nc{\w}{\omega}
\nc{\io}{\iota}
\nc{\z}{\zeta}
\nc{\s}{\sigma}
\nc{\Si}{\Sigma}
\nc{\vphi}{\varphi}
\nc{\bP}{\bar{P}}
\nc{\bQ}{\bar{Q}}

\nc{\ran}{\rangle}
\nc{\lan}{\langle}

\newcommand{\one}{\mathbf{1}}

\nc{\bfone}{{\bf 1}}
\nc{\1}{{\bf 1}}

\renewcommand{\Re}{\mathrm{Re}} 
\renewcommand{\Im}{\mathrm{Im}} 
\newcommand{\re}{\operatorname{Re}}
\newcommand{\im}{{\rm Im}}

\newcommand{\p}{\partial}

\newcommand{\n}{\nabla}

\newcommand{\DETAILS}[1]{}
\renewcommand{\st}[1]{}

\begin{document}

\rightline {\small \emph{Published in}}  
\rightline{\small \emph{Lett Math Physics, 2026}}
\bigskip

\title[On light cone bounds for Markov quantum open systems]{On light cone bounds for Markov quantum open systems}
\author[I.~M.~Sigal]{Israel Michael Sigal}
	\address{Department of Mathematics, University of Toronto, Toronto, ON M5S 2E4, Canada }
	\email{im.sigal@utoronto.ca}
	
	\author[X.~Wu]{Xiaoxu Wu}
	\address{Mathematical Sciences Institute, Australian National University, Acton ACT 2601, Australia}
		\email{Xiaoxu.Wu@anu.edu.au}
\date{\today}


\maketitle
{
}

\tableofcontents
\begin{abstract}
  We study space-time behaviour of solutions of the von Neumann-Lindblad equations underlying the dynamics of  
   Markov quantum open systems.  For a large class of these equations, we prove the existence of an effective light cone with an exponentially small spill-over.

\end{abstract}


 \section{Introduction}\label{ref: setup}
\subsection{Markovian quantum open systems} 
\label{set: 1.1} In this paper, we study space-time dynamics of Markov open quantum systems (MOQS) on  the Hilbert space $\mathcal H=L^2(\Lambda)$, where $\Lambda$ is either $\mathbb R^n$ or $\mathbb Z^n$. 
%
 %
We prove the existence of an effective light cone with an exponentially small spill-over for a large class of 
such systems.


{\red Conceptually, our bounds are related to the celebrated Lieb-Robinson bound {\red\cite{LR}} which plays a central role in analysis of evolution of quantum information (see e.g. {\red \cites{NachVerZ, Pou} for the Lindladian setting in quantum spin systems and \cites{BH,BHV,CL,CL2,EisOsb,EldredgeEtAl,EpWh, FLS2, FLSZ, H1,H0,H2,H3,HastKoma, KGE, KuwLem2024,  KuwSaito1,KVS, LRSZ, NachOgS, NRSS, NachOgS2,NachSchlSSZ,NachS,NSY2, RS, SHOE, SZ,TranEtAl3,TranEtal5} for the unitary evolutions)}.}


An open quantum system (OQS) is a pair $(\mathcal S_1^+, \beta_t)$, where the state space $\mathcal S_1^+$ is the space of positive trace-class operators (density operators) on a Hilbert space $\mathcal H$ and the evolution $\beta_t$ is a family of quantum maps (or quantum channels) on $\mathcal S_1^+$, i.e. linear, completely positive, trace preserving maps (see \cites{Al, AlLe, BrPet, Davies4, GS, Kr1} and the references therein).\par

The concept of OQS is an extension of that of the (closed) quantum system ($\mathcal S_1^+, \alpha_t$), where $\alpha_t$ is the von Neumann dynamics,
\begin{equation}\label{eq: 1.1}
\alpha_t(\rho)=e^{-iHt}\rho e^{iHt},
\end{equation}
incorporating, in a natural way, the influence of the system's environment. OQS arise also in the quantum measurement theory where the degrees of freedom of systems under investigation (rather than of the environment) are integrated out and in studying entanglement between two or more systems.

Importantly, even when the interaction with  environment can be neglected, investigation of OQS is needed to determine whether properties of closed systems are robust w.r.to weak interaction with an outside environment. For instance, whether transmission of quantum information is stable w.r.to decoherence induced by such an interaction.\par
It is shown in \cites{GKS, Lind} that under the Markovian assumption that $\beta_t$ is a strongly continuous semigroup, 
 \begin{equation}\label{beta-t-semigr}
  \beta_t \circ \beta_s= \beta_{t+s},\, \qquad\forall t,s\ \geq 0, \quad\text{ and }\beta_t\xrightarrow{s} \mathbbm 1\text{ as }t\downarrow 0,
\end{equation}
the evolution $\rho_t=\beta_t (\rho_0)$ satisfies the von Neumann-Lindblad equation (vNLE) (here and in the rest of this paper, we set $\hbar =1$) 
\begin{equation}\label{vNLeq}\p_t \rho_t= -i[H, \rho]+\sum_{j=1}^{\infty}(W_j \rho_t W_j^*-\frac{1}{2}\{W_j^* W_j,\rho_t\}),
\end{equation}
with the initial condition $\rho_{t=0}=\rho_0$. Here $H$ and $W_j, j=1,\cdots$, are operators on $\mathcal H$, $H$ is a quantum Hamiltonian of the system of interest and $W_j$ are operators produced by the interaction with environment, called the jump operators, and  $\{A,B\}:=AB+BA$.

Conversely, under rather general conditions (see a discussion below and in Appendix \ref{sec:vNLEexist}), solutions to the vNLE exist for any initial condition $\rho_0$ in $\mathcal S_1$ and generate Markov open quantum (MOQ) dynamics, $\beta_t(\rho_0)=\rho_t$. Thus the class of MOQ semigroups is rather rich.  Furthermore, equations of the form~\eqref{vNLeq} were derived in the van Hove limit of a particle system coupled to a thermal reservoir, see 
\cites{Davies1,Davies2,Davies3,Davies4, JP1}. Hence,~\eqref{vNLeq} captures, at least approximately, natural physical models.


Clearly, the vNLE is an extension of the von Neumann equation (vNE)
\begin{equation}\label{vNE}
\p_t \rho_t=-i[H, \rho_t],
\end{equation} 
which generates evolution~\eqref{eq: 1.1}, describing the statistics of closed quantum systems. 
While the vN dynamics can be always reduced to the Schr\"odinger one on the corresponding Hilbert space ($L^2(\Lam)$, in our case), this is not true for the vNLE.  
Thus, vNLE is a genuine extension of  the Schr\"odinger equation beyond QM (to OQS, or to what can be termed as quantum statistics) making it 
  a central object of quantum physics.


By virtue of its origin, the vNLE plays a foundational role 
 in quantum information science and in non-equilibrium quantum  statistical mechanics,
 \DETAILS{ to study the evolution of quantum information and 
 the thermalization 
   and decoherence,} 
   see \cites{CEPH, CSM, DCMB, 
   IngKoss, IngKoss2, Pou,  
   Weinb}, and \cites{Al, AlLe, OS}, respectively, and references therein. It is also used in computational physics to construct the Gibbs and ground states for given Hamiltonians,~\cites{CB, CKB, CL, DCL, KB, SM, TOVPV, YA} and references therein. See  \cite{Presk} for an elementary lecture-notes exposition of a role of vNLE in quantum information theory.
   
The vNLE also appears naturally in 
Fr\"ohlich et al theory of randomness in Quantum Mechanics (ETH-Approach, see \cite{FGP} and references therein).

Mathematically, vNLE  is a key representative of non-abelian PDEs. It is related to stochastic differential equations on Hilbert spaces, see~\cite{Holevo}.

As is standard, we assume that the operators $H$ and $W_j, j=1,\cdots$, satisfy the conditions 
\begin{itemize}
\item[(H)] $H$ is a self-adjoint operator; 
\end{itemize}
\begin{itemize}
\item[(W)] $W_j, j=1,\cdots,$ are bounded operators s.t.\\ $\sum\limits_{j=1}^\infty W_j^* W_j$ converges weakly.
\end{itemize}
\par

It is shown in \cite{Davies5} 
 that, under conditions (H) and (W),  the operator 
\begin{equation}\label{L}
 L(\rho)=-i[H, \rho]+\sum_{j=1}^{\infty}(W_j \rho W_j^*-\frac{1}{2}\{W_j^* W_j,\rho\}),
\end{equation}
 defined by the r.h.s. of 
vNLE, generates a OQD semigroup, $\beta_t=e^{L t}$. This implies, in particular, 
that Eq.~\eqref{vNLeq} with initial conditions in $\mathcal S_1^+$ has unique weak 
solutions in $\mathcal S_1^+$ (and strong solutions on the natural domain of $L$), see \cite{OS} for a streamlined version and more references, and Appendix~\ref{sec:vNLEexist} below, for a brief discussion. 

 We call a QOD $\beta_t$ satisfying~\eqref{beta-t-semigr}, the Markov QOD, or MQOD and $L$, the von Neumann-Lindblad (vNL) generator. \par

For other results on vNLE~\eqref{vNLeq}, we mention the scattering theory, see \cites{FFFS, FaFr}, and the problem of return to equilibrium, see~\cite{OS} and references therein. \par


\underline{Notation} In what follows, $\mathcal H=L^2(\Lambda)$, where $\Lambda$ is either $\mathbb R^n$ or $\mathbb Z^n$, $\mathcal B(X)$ denotes the space of bounded operators on a Banach space $X$, and $\mathcal S_1$ and $\mathcal S_2$, the Schatten spaces of trace-class and Hilbert-Schmidt operators. The norms in $\mathcal H$ and $\mathcal B(\mathcal H)$ are denoted by $\|\cdot\|$, and in $\mathcal S_1, \mathcal S_2$ and $\mathcal B(\mathcal S_1)$, by $\|\cdot\|_1, \|\cdot\|_2$ and $\|\cdot\|_{1}^{op}$, respectively. Explicitly, $\|\lambda\|_1=\Tr(\lambda^*\lambda)^{\frac{1}{2}}$ and $\|\lambda\|_2=(\Tr\lambda^*\lambda)^{\frac{1}{2}}$. $A,B$ will denote bounded operators (observables), $X,Y\subset \Lambda$ stand for subsets of $\Lambda$ and $\chi_X$, the characteristic function of $X\subset \Lambda$. In what follows, $\lambda,\mu\in \mathcal S_1$ and $\rho\in \mathcal S_1^+$, always. \par

To fix ideas, we assume that the DO's $\rho$ are normalized as $\Tr \rho=1$. 

\subsection{Light cone 
 bound} \label{sec:MVB-oqs}

Consider on $\mathcal H=L^2(\Lambda)$ the $n$-parameter group of unitary operators $T_\xi$ of multiplication by the function $e^{-i\xi\cdot x}, \xi \in \mathbb R^n$. \par

Define the polystrip $S_a^n, a>0,$ in the complex space $\mathbb C^n$ as 
\begin{equation}
    S_a^n:=\{ \zeta=(\zeta_1,\cdots,\zeta_n)\in \mathbb C^n\,:\, |\im \zeta_j|<a\quad \forall j\}.
\end{equation}
We assume the following conditions:
\begin{itemize}
\item[(AH)] The operators $H_\xi:=T_\xi HT_\xi^{-1}, \xi \in \mathbb R^n$, have the common domain $\mathcal D(H)$ and $H_\xi(H+i)^{-1}$ are bounded operators for all $\xi\in \mathbb R^n$ and the operator function $\xi\to H_{\xi}(H+i)^{-1},$ from $\mathbb R^n$ to $\mathcal B(\mathcal H)$, has an analytic continuation in $\xi$ from $\mathbb R^n$ to $S_a^n$ and this continuation, $H_\zeta$, is such that 
\begin{equation}\label{con: AH}
\text{ $\im H_{i\eta}:=\frac{1}{2i}(H_{i\eta}-H_{i\eta}^*)$ is a bounded operator $\forall\, i\eta\in S_a^n,\,$ $|\eta|=\mu$.}
\end{equation}

\item[(AW)] The operator-functions $\xi\to W_{j,\xi}=T_\xi W_jT_\xi^{-1}, j=1,\cdots,$ have analytic continuations, $W_{j,\zeta}$, as bounded operators from $\mathbb R^n$ to $S_a^n$ and these continuations satisfy (W) $\forall\zeta\in S_a^n$. 
\end{itemize}

For any two sets $X$ and $Y$ in $\Lambda$, let  $d_{XY}$ denote the distance between $X$ and $Y$ and  define $\hat\chi_X: \mathcal S_1\to \mathcal S_1$ by 
\begin{equation}
    \hat \chi_X(\rho)=\chi_X\rho\chi_X.\label{eq: 4.6}
\end{equation}

\begin{theorem}\label{thm1} Assume Conditions~(H),~(W),~(AH) and~(AW). Then, for any $\mu\in (0,a)$ and 
for any two disjoint sets $X$ and $Y$ in $\Lambda$, the MQOD $\beta_t$ satisfies
\begin{equation}
    \|\hat \chi_X\beta_t\hat \chi_Y\|_{1}^{op}\leq Ce^{-2\mu(d_{XY}-ct)},\label{eq: 4.7}
\end{equation}
for any $c>c(\mu)$ and some constant $C=C_{n,c,\mu}>0$ depending on $n,c,\mu$. Here  $c(\mu)\in (-\infty,\infty)$, is given  by~\eqref{def: cmu} below.
\end{theorem}

This theorem is proven in Section~\ref{sec: 2}. {\red By Proposition~\ref{prop: pos cnu} of Appendix~\ref{sec: positivity}, $c(\mu)\geq0$, with $c(\mu)>0$, $ \forall \mu\in (0,a)$, under the condition $L'(x)\neq0$. Here $L'$ is the Heisenberg-Lindblad generator, which is dual to the von Neumann-Lindblad one,~\eqref{L}, and $v=L'(x)$ is the Heisenberg-Lindblad velocity (\cites{Breteaux_2022,Breteaux_2023}), see~\eqref{def: L'} and Appendix~\ref{sec: positivity} below, which is well defined by Lemma~\ref{lem: cnu0} of Appendix~\ref{sec: positivity}. }\st{We conjecture that $c(\mu)>0$. Below, we show this under additional conditions on $H$.} 

For a set $X\subset \Lambda$, let $X^c:= \Lambda-X$. We say that a state $\rho$ is localized in $X$ if in $\rho$, the probability of the system to be in $X$ is equal to $1$:
\begin{align}\label{st-loc}\rho (  \chi_{X})\equiv \Tr( \chi_{X}\rho)=1 \quad \text{ or }\ \quad \rho (  \chi_{X^c})\equiv \Tr( \chi_{X^c}\rho)=0.\end{align}

\begin{corollary}\label{cor 1.2} Assume Conditions~(H),~(W),~(AH) and~(AW). Then, for any $\mu\in (0,a)$ and 
 for any $X,Y\subset \Lambda$ and any DO $\rho$ localized in $X$, the MQOD $\beta_t$ satisfies 
\begin{equation}\label{thm1: eq1}
    \Tr(\chi_Y\beta_t(\rho))\leq Ce^{-2\mu(d_{XY}-ct)}\Tr( \rho)
\end{equation}
for any $c>c(\mu)$ and some constant $C=C_{n,c,\mu}>0$ depending on $n,c,\mu$. 
\end{corollary}\par

For $\Lambda=\mathbb R^n$, the main examples of the quantum Hamiltonian we consider are given by operators of the form
 \begin{align} \label{H}H=\om(p)+V(x)\end{align} 
 acting on $L^2(\mathbb R^n)$. Here $\om(\xi)$ is a real, smooth, positive function on $\mathbb R^n,$ $p:=-i\n$ is the momentum operator 
  and  the potential $V(x)$ is real and $\om(p)$-bounded with the relative bound $<1$, i.e.
\begin{align} \label{V-cond}
& \exists \,0\le a <1,\ b>0: \quad
\|Vu\|\le a  \| \om(p) u\|+ b\|u\|.
\end{align}
These assumptions ensure that $H$ is self-adjoint on the domain of $\om(p)$. \par
Operator~\eqref{H} satisfies (AH) if the function $\omega(k)$ has an analytic continuation, $\omega(\zeta),$ from $\mathbb R^n$ to $S_a^n$ and $\im\, \omega(i\eta)$ is a bounded function $\forall i\eta\in S_a$.\par

An important example of $\omega(k)$ is the relativistic dispersion law $\omega(k)=\sqrt{|k|^2+m^2}$ with $m>0$, or more generally, $\omega(k)=\sum\limits_{j=1}^N \sqrt{|k_j|^2+m_j^2},$ with $k=(k_1,\cdots,k_N), k_j\in \mathbb R^d, m_j>0$. Thus conditions (H) and (AH) are satisfied for the semi-relativistic $N$-particle quantum Hamiltonian (cf.~\cite{SigWu})
\begin{equation}
H=\sum\limits_{j=1}^N \sqrt{|p_j|^2+m_j^2}+V(x),
\end{equation}
where $x=(x_1,\cdots,x_N), x_j\in \mathbb R^d$, and $p_j=-i\nabla_{x_j}$, $j=1,\cdots,N$, and $V(x_1,\cdots,x_N)$ is a standard $N$-body potential. Hence Theorem~\ref{thm1} holds for semi-relativistic $N$-body systems.\par

For $\Lambda=\mathbb Z^n$, an example of the operator $H$ is given by 
\begin{equation}\label{def: HTV}
H=T+V(x),
\end{equation}
where $T$ is a symmetric operator and $V(x)$ is a real, bounded function.\par

Furthermore, Condition~(HA) says that $T$ has exponentially decaying matrix elements $t_{ { x, y}}$, i.e.
\begin{equation}
  |t_{{x, y}}|\leq Ce^{-a|{ x - y}|}, \qquad \text{for some }a>0,
\end{equation}
e.g. the discrete Laplacian $\Delta_{\mathbb Z^n}$ on $\mathbb Z^n$. \par

There are no canonical physical models for $\{W_j\}_{j=1}^{j=\infty}$. Any family of operators $\{W_j\}_{j=1}^{j=\infty}$ satisfying (W) (and (AW) whenever needed) is acceptable.\par

For $\Lambda=\mathbb Z^n$, {the operator-family} $T_\xi$ in Condition~(A) {depends on the $\mathbb Z^n$-equivalence classes of $\xi$'s varying in} the dual (quasimomentum) space $K\equiv  
\mathbb R^n/ \mathbb Z^n$, and $\xi \cdot x$ {could be thought of as a} linear functional on $K$. (For a general lattice $\mathcal L$ in $\mathbb R^n$, the (quasi) momentum space $\mathcal L^*$ is isomorphic to the torus $\mathbb R^n/ \mathcal L'$, where $\mathcal L'$ is the lattice reciprocal to $\mathcal L$.) Furthermore, the strip $S_a^n$ (see Condition A) {could be identified with} $\{\zeta\in {K+i \R^n}\,:\, |{\Im} \zeta_j|<a \, \forall\, j\}$. \par

The second key ingredient in the quantum theory is the notion of observables. Though physical observables are self-adjoint, often unbounded, operators on $\mathcal H$ representing actual physical quantities (say, $p=-i\nabla$ for $\Lambda=\mathbb R^n$), it is convenient mathematically to consider as observables all bounded operators $A\in \mathcal B(\mathcal H)$.\par 

An average of a physical quantity (say, momentum) represented by an observable $A$ in a state $\rho$ is given by $\Tr(A\rho)$. There is a duality between states and observables given by the coupling 
\begin{equation}\label{def: Arho}
  {{ (A,\rho)\equiv \rho(A):=}}\Tr(A\rho),\qquad \forall \, A\in \mathcal B(\mathcal H)\, \text{ and }\, \rho\in \mathcal S_1^+,
\end{equation}
which can be considered as either a linear, positive functional of $A$ or a convex one of $\rho$. In what follows, we use the notation
\begin{equation}
\rho (A):=\Tr( A\rho).
\end{equation}
Here $A\to \rho(A)$ is a linear positive functional on the Banach space, in fact, $C^*$-algebra, $\mathcal B(\mathcal H)$.\par

By the duality,~\eqref{def: Arho}, the von Neumann dynamics yields the Heisenberg one, while the von Neumann-Lindblad dynamics $\beta_t$ of states produces the dynamics $\beta_t'$ of observables as
\begin{equation}\label{def: beta'}
\Tr(\beta_t'(A)\rho)=\Tr(A\beta_t(\rho)).
\end{equation}
Under the Markov assumption~\eqref{beta-t-semigr}, the dynamics $\beta_t'$ has the weak Markov property
\begin{equation}
    \beta_s'\circ \beta_t'=\beta_{s+t}', \qquad \forall\, s,t \geq0, \text{ and }\beta_t'\xrightarrow{w} \mathbbm 1 \text{ as }t\to 0,
\end{equation}
and $A_t=\beta'_t(A)$ is weakly differentiable in $t$ and weakly satisfies the dual Heisenberg-Lindblad (HL) equation (see \cite{OS} and the references therein) 
\begin{equation}
    \p_t A_t=i [H, A_t]+\sum\limits_{j=1}^\infty (W_j^* A_tW_j-\frac{1}{2}\{ W_j^* W_j, A_t\}).\label{eq: ptAt}
\end{equation}
In fact, this equation has a unique strong solution for any initial condition from a dense set in $\mathcal B(\mathcal H)$ (see e.g.~\cite{OS} and Remark~\ref{rem: HL} below).
\begin{theorem}\label{thm:HL-mvb} 
Assume Conditions~(H),~(W),~(AH) and~(AW). Then, for any $\mu\in (0,a)$ and 
 for any two disjoint sets $X$ and $Y$ in $\Lambda$, the dual MQOD $\beta'_t$ satisfies
    \begin{equation}\label{HL-mvb} 
            \|\hat \chi_X\beta'_t\hat \chi_Y\|\leq Ce^{-2\mu(d_{XY}-ct)},
    \end{equation}
    for any $c>c(\mu)$ and some constant $C=C_{n,c,\mu}>0$ depending on $n,c,\mu$.   Here, recall,   $c(\mu)$ is given  by~\eqref{def: cmu}.
\end{theorem}

\begin{lemma} Theorem~\ref{thm:HL-mvb} is equivalent to Theorem~\ref{thm1}. 
\end{lemma}
\begin{proof}  Theorem~\ref{thm1} and the relation (see \cite{Schatten}, Chapter IV, Section 1, Theorem 2)
\begin{equation}\label{eq: norm-A}
    \|A\|=\sup\limits_{\rho\in \mathcal S_1^+,\, \Tr\rho=1}|\Tr(A\rho)|
\end{equation}
imply Theorem~\ref{thm:HL-mvb}. In the opposite direction, Theorem~\ref{thm:HL-mvb} and the relation
\begin{equation}
    \|\lambda\|_1=\sup\limits_{A\in \mathcal B,\, \|A\|=1}|\Tr(A\lambda)|\label{easy-inequal}
    \end{equation}
proven below, imply Theorem~\ref{thm1}. To prove~\eqref{easy-inequal}, we notice that, by the polar decomposition, $\|\lambda\|_1=\Tr(\lambda U)$ for every $\lambda\in \mathcal S_1$ and some unitary operator $U$, we have $\|\lambda\|_1\leq \sup\limits_{A\in \mathcal B,\, \|A\|=1} |\Tr(A\lambda)|.$ On the other hand, we have the standard inequality
\begin{equation}
   |\Tr(A\lambda)|\leq \|A\|\|\lambda\|_1.
\end{equation}
These two relations imply~\eqref{easy-inequal}.\end{proof}

  \begin{remark}\label{rem: HL}
The HL generator $L'$ on the r.h.s. of~\eqref{eq: ptAt} can be written as 
    \begin{equation}\label{def: L'}
    L'A=i[H,A]+\psi'(A)-\frac{1}{2}\left\{\psi'(\mathbbm1), A\right\},
    \end{equation}
    where $\psi'$ is a completely positive map on $\mathcal B$, which, by the Krauss' theorem, is of the form 
    \begin{equation}\label{psi'}    \psi'(A)=\sum\limits_{j=1}^\infty W_j^*A W_j, 
    \end{equation}
    for some bounded operators $W_j,j=1,2,\cdots,$ satisfying (W). 
\end{remark}
This representation allows for an easy proof of existence of mild and strong solutions to Eq.~\eqref{eq: ptAt}. Indeed,~\eqref{eq: ptAt} can be written as $\partial_t A_t=L'A_t$, with the operator $L'$ given by~\eqref{def: L'}. Furthermore, $L'$ can be written as $L'=L_0'+G'$, where $L_0'A=i[H,A]$ and 
\begin{equation}\label{G'}
G'(A):=\psi'(A)-\frac{1}{2}\{\psi'(\mathbbm1), A \}.
\end{equation}
Now, we show boundedness of the map $G'$. Indeed, the operator $\psi'(\mathbbm1)=\sum\limits_{j=1}^\infty W_j^*W_j$ is bounded, by Condition~(AW), and positive. Furthermore, the map $\psi'$ is positive and therefore $\| \psi'(A)\|\leq \psi'(\mathbbm1) \|A\|$, for any self-adjoint operator $A$, which follows by applying $\psi'$ to the operator $B=\|A\|\mathbbm1-A\geq 0$. One can extend this bound to non-self-adjoint operators to obtain 
\begin{equation}
\|G'(A)\|\leq 2 \psi'(\mathbbm1)\|A\|.\label{ineq: G'A}
\end{equation}\par By the explicit representation $e^{L_0't}A=e^{iHt}Ae^{-iHt}$, the operator $L_0'$ generates a one-parameter group $\alpha_t'=e^{tL_0'}$ of isometries on $\mathcal B$ (Heisenberg evolution), and therefore, since $G'$ is bounded, by a standard perturbation theory, $L'$ generates a one-parameter group of bounded operators, $\beta_t'=e^{tL'}$. \par

Since $\beta_t'=e^{L't}$ is (completely) positive (by the original assumption on $\beta_t$) and unital ($e^{L't}\mathbbm1=\mathbbm1$), as follows from $L'\mathbbm1=0$, we have $\|e^{L't}\|\leq 2$. (In the opposite direction, one can prove the complete positivity of $e^{L't}$ by using Eq.~\eqref{def: L'}, see~\cite{OS} and references therein.)\par
\begin{remark}
    Using~\eqref{def: L'}, one can formulate the HLE on an abstract von Neumann algebra with unity. We expect that Theorem~\ref{thm:HL-mvb} can be extended to this setting with $x$ replaced by a self-adjoint affiliated with the algebra. 
\end{remark}

\begin{remark}If one thinks of the algebra of observables $\mathcal B\equiv\mathcal B(\mathcal H)$ and the Heisenberg (resp. Heisenberg-Lindblad) dynamics on it as primary objects, then one might define the state space as the dual $\mathcal B'$ of $\mathcal B$ with the dynamics given by the von Neumann (resp. von Neumann-Lindblad) dynamics. Then $\mathcal S_1$ is a proper, closed subspace of $\mathcal B'$ (see~\cite{Schatten}, Chapter IV, Theorems 1 and 5) invariant under the von Neumann and von Neumann-Lindblad dynamics. By restricting the von Neumann dynamics further to the invariant subspace of $\mathcal S_1$ of rank $1$ orthogonal projections one arrives at a formulation equivalent to the standard quantum mechanics. For closed systems, the latter extends uniquely to von Neumann dynamics on $\mathcal S_1$ and then on $\mathcal B'$. For open systems, this is not true any more: the minimal state space for the vNL dynamics is $\mathcal S_1$.   
\end{remark}

\subsection{Comparison with earlier results and description of the approach}

Bounds of the form of~\eqref{thm1: eq1} but with a power decay were obtained in~\cites{Breteaux_2022,Breteaux_2023}.  {\red (Note that ~\cites{Breteaux_2022,Breteaux_2023} treat a more general class of vNLE including non-analytic ones (long-range in the condensed matter physics terminology) for which exponential bounds do not hold.)} For the von Neumann evolution, \eqref{vNE}, where 
the key estimates reduce to estimating the Schr\"odinger unitary, $e^{-iHt}$, a result similar to Theorems~\ref{thm1} 
 was proven in~\cite{SigWu}. 




Presently, there are three approaches to proving light-cone estimates.  The first approach going back to Lieb and Robinson (see~\cite{NachOgS2} for a review) is based on a perturbation (Araki-Dyson-type) expansion.\par

In the second approach, one constructs special observables (adiabatic, space-time, local observables or ASTLO) which are monotonically decreasing along the evolution up to self-similar and time-decaying terms (recursive monotonicity). Originally designed for the scattering theory in quantum mechanics in \cite{SigSof2} and extended in \cites{AFPS,APSS,BoFauSig, FrGrSchl, HeSk, HunSigSof, SchSurvey, Sig, Skib}, this approach was developed in the many-body theory context (\cites{FLS1, FLS2,LRSZ, SZ}) proving light-cone bounds on the propagation in bose gases, the problem which was open since the groundbreaking work of Lieb and Robinson (\cite{LR}) in 1972.\par


In this paper, we develop 
the third approach, initiated in~\cite{SigWu} (see also \cites{CJWW, FLSZ}). 
Specifically, we reduce the problem of proving space-time estimates on solutions to vNLE to constructing analytic deformations of the evolution $\beta_t=e^{Lt}$ and estimating these deformations as well as the geometrical factors $\chi_U(x) e^{-i\z\cdot x}$ for $\z\in S_a^n$ and various domains $U\subset \Lam$.  {\red (We also use the Araki-Dyson expansion, but we do this for a technical step of showing that the deformed von Neumann-Lindblad operator generates a semigroup, which probably could be bypassed.)}



In the process, we construct a theory of analytic deformations of the vNLE (or $\beta_t$) and expand the analytical toolbox for dealing with maps on operator spaces. \st{including  estimates on generalizations of completely positive maps of the form} 
\st{introduced in this paper, which, for want of a better term, we call sub-completely positive maps.}\par

This paper is organized as follows. In Section~\ref{sec: 2}, we prove Theorem~\ref{thm1}, modulo two propositions which are proven in Section~\ref{sec: pf prop}, after we demonstrate some inequalities for completely positive (quantum) and related maps in Section~\ref{sec: 3}. 
 In Appendix~\ref{sec:vNLEexist}, we sketch an existence theory for vNLE. 
  Section~\ref{sec: 2} could also serve as a sketch of the proof of Theorem~\ref{thm1}.

\section{Proof of Theorem~\ref{thm1} given Propositions~\ref{prop 4.1} and~\ref{prop: 2.9}}\label{sec: 2}

Recall our convention that $\lambda, \mu\in \mathcal S_1$ and $\rho\in \mathcal S_1^+$. For $\xi,\eta\in \mathbb R^n$, we let $T_{\xi,\eta}\lambda=T_\xi \lambda T_\eta^{-1}$, with the 'left' and 'right' sides of $\lambda$ treated differently, $L_{\xi,\eta}=T_{\xi,\eta} L T_{\xi,\eta}^{-1} $ and $ \beta_{t,\xi,\eta}:=T_{\xi,\eta} \beta_t T_{\xi,\eta}^{-1}$. Since $T_\xi$ is a unitary group (of multiplication operators by $e^{-i\xi\cdot x}$) on $L^2(\Lam)$, $T_{\xi,\eta}$ is a group of isometries on $\mathcal S_1$ and $L_{\xi,\eta}$ and $\beta_{t,\xi,\eta}$ are isometric deformations of $L$ and $\beta_t$. Furthermore, we define $\hat R \lambda=(H+i)^{-1}\lambda(H-i)^{-1}$ and $\hat R(\mathcal S_1)=\{\hat R(\lambda)\,:\, \lambda\in \mathcal S_1\}$.

We assemble all technical results needed in the proof of Theorem~\ref{thm1} in the following proposition proven  in Section~\ref{sec: pf 2.1}.

\begin{prop}\label{prop 4.1} Assume Conditions~(H),~(W),~(AH) and~(AW). Consider the family operators $L_{\zeta, \theta}$ on $\hat R(\mathcal S_1)$ of the form 
\begin{equation}
    L_{\zeta,\theta}=L_{0,\zeta, \theta}+G_{\zeta, \theta},  \label{eq.form}
    \end{equation}
    with the operators $L_{0,\zeta, \theta}$ and $G_{\zeta, \theta}$ given by
    \begin{equation}\label{eq: L0zeta}
    L_{0,\zeta, \theta}\lambda:=-i(H_\zeta \lambda-\lambda H_ \theta),
    \end{equation}
    \begin{equation}
        G_{\zeta, \theta}\lambda:=\sum\limits_{j=1}^\infty  \left( W_{j,\zeta}\lambda W_{j,\bar{ \theta}}^*-\frac{1}{2}W_{j,\bar{\zeta}}^*W_{j,\zeta}\lambda-\frac{1}{2}\lambda W_{j,\bar{ \theta}}^*W_{j, \theta}\right),\label{eq.Grho}
    \end{equation}
where $H_\zeta,$ $W_{j,\zeta}$ and $W_{j,\bar \zeta}^*\equiv(W_j^*)_\zeta$ are analytic continuations of $H_\xi$, $W_{j,\xi}$ and $W_{j,\xi}^*$. Then, we have the following statements:\par

(a) $L_{0,\zeta, \theta}$ and $G_{\zeta, \theta}$ are bounded maps from $\hat R(\mathcal S_1)$ to $\mathcal S_1$ and on $\mathcal S_1$, respectively, $\forall \zeta, \theta\in S_a^n$.

(b) $L_{\zeta, \theta}\hat R$ is an analytic continuation (as a family of bounded operators) of $L_{\xi,\tilde \xi}\hat R$  from $\mathbb R^n\times \mathbb R^n$ to $S_a^n\times S_a^n$.

(c) $L_{\zeta, \theta}$ generates the family of bounded one-parameter groups $\beta_{t,\zeta, \theta}=e^{L_{\zeta, \theta}t}$ and the latter family is analytic in $\zeta, \theta\in S_a^n$. 

    (d)  $\beta_{t,\zeta,-\zeta}$ is positivity preserving on $\mathcal S_1$.     
\end{prop}
   

Using that $T_{-\eta}=T_\eta^{-1}$ and $ T^{-1}_{\xi,\eta}=T_{-\xi,-\eta}$, we write $T_{\xi,\eta}\lambda:=T_\xi\lambda T_\eta^{-1}=T_\xi\lambda T_{-\eta}$ and $\beta_{t,\xi,\eta}:= T_{\xi,\eta} \beta_t T_{-\xi,-\eta} $. We have  
\begin{lemma} For any bounded sets $U,V\subset \Lambda$ and any $\zeta,  \theta\in S_a^n$, we have
\begin{equation}
     \hat \chi_{U} \beta_t \hat \chi_{V} =\hat \chi_{U}  T_{-\zeta, - \theta} \beta_{t,\zeta, \theta}T_{\zeta, \theta}\hat \chi_{V}, \quad\text{$\zeta, \theta\in S_{a}^n$}.\label{eq: XbtY0}
\end{equation} 
\end{lemma}
\begin{proof} Using that $T_{-\xi,-\eta}T_{\xi,\eta}=\mathbbm1$, we write
\begin{align}
    \hat \chi_U \beta_t \hat \chi_V =&\hat \chi_U  T_{-\xi,-\eta} T_{\xi,\eta} \beta_{t} T_{-\xi,-\eta}  T_{\xi,\eta}\hat \chi_{V}\nonumber\\
    =&\hat \chi_U T_{-\xi,-\eta} \beta_{t,\xi,\eta}T_{\xi,\eta}\hat \chi_{V}.
\end{align}
By Conditions~(AH) and~(AW) and the facts, that for $U$ and $V$ bounded, the operators $\hat \chi_U T_{-\zeta, - \theta}$ and $T_{\zeta, \theta}\hat \chi_V$ are bounded and analytic for $\zeta, \theta\in S_a^n$, we can continue the right-hand side analytically in $\xi$ and $\eta$ from $\mathbb R^n$ to $\mathcal S_{a}^n$ to obtain~\eqref{eq: XbtY0}.\end{proof}

\eqref{eq: XbtY0} is our key relation, a basis of our estimates. The idea behind this relation is related to the Combes-Thomas argument (\cite{AW}, see~\cite{CT}, for a book presentation and extensions).\par

Now, we estimate~$\hat \chi_U \beta_t\hat \chi_V $ for $U$ and $V$ arbitrary disjoint sets in $\Lambda$ and then, using partitions of unity, we obtain the desired estimate of $\hat \chi_X \beta_t\hat \chi_Y$.\par

 We denote $\beta_{t,\zeta}:=\beta_{t,\zeta,-\zeta}$. Using the relation  
\begin{equation}
    \hat \chi_U T_{\zeta,-\zeta}(\lambda)=\chi_UT_{\zeta} \lambda T_{\zeta}\chi_U,
\end{equation}
we estimate 
\begin{equation}\label{eq: 7.5}
    \|\hat \chi_U T_{\zeta,-\zeta}(\lambda)\|_{1}\leq \| \chi_UT_{\zeta}\|^2 \|\lambda\|_{1}, \text{ for any }\lambda\in \mathcal S_1.
\end{equation}
Using~\eqref{eq: XbtY0}, together with~\eqref{eq: 7.5}, we obtain
\begin{equation}
\begin{aligned}
    \|\hat \chi_U\beta_t\hat\chi_{V}(\lambda) \|_{1}=& \| \chi_UT_{-\zeta}\left(\beta_{t,\zeta}\hat\chi_{V} T_{\zeta,-\zeta}(\lambda) \right) \chi_UT_{-\zeta}\|_{1}\\
    \leq& \|\chi_U T_{-\zeta}\|^2 \|\beta_{t,\zeta} \|_{1}^{op}\| \hat \chi_V T_{\zeta,-\zeta}\lambda\|_{1}\\
    \leq & \|\chi_U T_{-\zeta}\|^2\|\chi_V T_{\zeta}\|^2 \|\beta_{t,\zeta} \|_{1}^{op} \| \lambda\|_{1},
\end{aligned}
\end{equation}
where, recall, $\|\cdot\|_{1}^{op}$ denotes the norm of operators on $\mathcal S_1$, which implies 
\begin{equation}
    \|\hat \chi_U \beta_t \hat \chi_V\|_{1}^{op}\leq \|\chi_U T_{-\zeta}\|^2\|\chi_V T_{\zeta}\|^2 \|\beta_{t,\zeta} \|_{1}^{op}.\label{eq: 7.8}
\end{equation}

Now we estimate the norms on the r.h.s. of~\eqref{eq: 7.8} beginning with $\|\beta_{t,\zeta} \|_1^{op}$. For a self-adjoint operator $A$, we denote 
\begin{equation}
\sup A=\sup\limits_{\psi\in \mathcal D(A), \, \|\psi\|=1}\langle\psi, A\psi \rangle.
\end{equation}
\begin{prop}\label{prop: beta acbd} 
Let $\zeta=i\eta, \eta\in \mathbb R^n$, $|\eta|=\nu$, $\nu\in (0,a)$. Then we have
    \begin{equation}
        \|\beta_{t,\zeta} \|_1^{op}\leq 4e^{2\nu c'(\nu)t},\label{beta acbd}
    \end{equation}
where the parameter function $c'(\nu)$ is given by 
\begin{equation}
c'(\nu):=\sup\limits_{\zeta=i\eta,|\eta|=\nu} \sup\left(\im H_\zeta+\tilde G_\zeta\right)/\nu.\label{def: velocity bd}
\end{equation}
Here $\tilde G_\zeta$ is the bounded, self-adjoint operator on $\mathcal H$ given by 
\begin{equation}
    \tilde G_\zeta:=\frac{1}{2}\sum\limits_{j=1}^\infty (W_{j,\zeta}^*W_{j,\zeta}-\frac{1}{2}W_{j,-\zeta}^*W_{j,\zeta}-\frac{1}{2}W_{j,\zeta}^*W_{j,-\zeta}),\label{def: tG}
\end{equation}
where $W_{j,\zeta}^*\equiv(W_{j,\zeta})^*=(W_j^*)_{\bar\zeta}$. Moreover, $c'(\nu)\in(-\infty,\infty)$.
\end{prop}
Observe that 
\begin{equation}\label{eq: tGzeta sys}
\tilde G_\zeta\geq \frac{1}{4}\sum\limits_{j} \left(W_{j,\zeta}^*W_{j,\zeta}-W_{j,-\zeta}^*W_{j,-\zeta}\right).
\end{equation}

\begin{proof}[Proof of Proposition~\ref{prop: beta acbd}] 
    Since the operator-functions $(W_j^*)_{\bar\zeta}$ and $(W_{j,\zeta})^*$ are analytic in $\bar \zeta$ and equal for $\zeta\in \mathbb R^n$, they are equal for all $\zeta\in S_a^n$. This yields
    \begin{equation}
        W_{j,\zeta}^*=(W_{j,\bar \zeta})^*.\label{4.6}
    \end{equation}    
    Hence, by~\eqref{def: tG}, $\tilde G_\zeta$ is formally symmetric. Moreover, by Condition (AW),~\eqref{def: tG}, \\~\eqref{4.6} and inequality 
    \begin{align}
        \|\sum\limits_{j=1}^\infty A_j^* B_j\|\leq &\|\sum\limits_{j=1}^\infty A_j^*A_j\|^{\frac{1}{2}}\|\sum\limits_{j=1}^\infty B_j^*B_j\|^{\frac{1}{2}}\nonumber\\
        \leq& \frac{1}{2}\left( \|\sum\limits_{j=1}^\infty A_j^*A_j\|+\|\sum\limits_{j=1}^\infty B_j^*B_j\|\right),\label{est: 2.38}
      \end{align}
    which follows by applying the Cauchy-Schwarz inequality to $|\langle \psi, \sum\limits_{j=1}^\infty A_j^*B_j\psi\rangle|$, the operator $\tilde G_{\zeta}$ is bounded and therefore self-adjoint. This and Condition (AH) imply that $c'=c'(\nu)<\infty$. Next, we need the following lemma.
    \begin{lemma}\label{lem: tGzeta} For $\zeta\in S_a^n, \re \zeta=0,$ and all $\rho\in \mathcal S_1^+$, we have, for $G_{\zeta, \theta}$ defined in~\eqref{eq.Grho},
    \begin{equation}
        \Tr(G_{\zeta,- \zeta}\rho)= \Tr(\tilde G_{\zeta} \rho).\label{est: reTrG}
    \end{equation}
    \end{lemma}
    \begin{proof} By~\eqref{eq.Grho}, 
     with $ \theta=-\zeta$ and $\re\zeta=0$, we have $\bar{ \theta}=\zeta$, $\bar\zeta=-\zeta$ and
    \begin{equation}
        \begin{aligned}
            & \Tr(G_{\zeta,- \zeta}\rho)\\
            =&  \Tr\sum\limits_{j=1}^\infty \left( W_{j,\zeta} \rho W_{j,\zeta}^*-\frac{1}{2} W_{j,-\zeta}^*W_{j,\zeta} \rho-\frac{1}{2}  \rho W_{j,\zeta}^*W_{j,-\zeta}\right).
        \end{aligned}\label{reTrGrho}
    \end{equation}
    This relation, together with the definition of $\tilde G_\zeta$ (see Eq.~\eqref{def: tG}) and the cyclicity of the trace, implies~\eqref{est: reTrG}. 
    \end{proof}
    Fix $\zeta=i\eta, \eta\in \mathbb R^n$ with $|\eta|=\nu, \nu \in (0,a)$. We write any $\lambda\in \mathcal S_1$ as 
    \begin{equation}
        \lambda=\lambda_+-\lambda_-+i(\lambda_+'-\lambda_-'), \quad \text{with }\lambda_\pm, \lambda_\pm'\in \mathcal S_1^+.\label{decom: lambda}
    \end{equation}
    Specifically, if $|\lambda|=\sqrt{\lambda^*\lambda}$, $\re \lambda=\frac{1}{2}(\lambda+\lambda^*)$ and $\im \lambda=\frac{1}{2i}(\lambda-\lambda^*)$, then $\lambda_\pm$ and $\lambda_\pm'$ are given by 
\begin{equation}
    \lambda_\pm:=\frac{|\Re \lambda |\pm\Re\lambda}{2}\quad \text{and}\quad \lambda_\pm':=\frac{|\im \lambda| \pm\im\lambda}{2}.
\end{equation}
Recall that $ \beta_{t,\zeta}(\lambda)\equiv\beta_{t,\zeta,-\zeta}(\lambda).$ By the linearity, it suffices to estimate $\beta_{t,\zeta}(\lambda)$ for $\lambda\in \mathcal S_1^+\cap \mathcal D(L_{\zeta,-\zeta})$. We note that by 
Proposition~\ref{prop 4.1}(d), $\lambda \in \mathcal S_1^+$ implies $ \beta_{t,\zeta}(\lambda) \in \mathcal S_1^+$. Hence, 
\begin{equation}
     \| \beta_{t,\zeta}(\lambda)\|_{1}=\Tr(\beta_{t,\zeta}(\lambda)),
\end{equation}
which yields, by Lemma~\ref{lem: tGzeta} and the relation $G_{\zeta,-\zeta}=2\tilde G_\zeta$, where $\tilde G_\zeta$ is defined in ~\eqref{def: tG}, and with $\lambda_{t,\zeta}:=\beta_{t,\zeta}(\lambda)$ and $\langle \lambda, \mu\rangle_{HS}:=\Tr(\lambda^* \mu)$, 
\begin{align}
    \p_t\| \lambda_{t,\zeta}\|_{1}=&\Tr(\p_t\lambda_{t,\zeta})\nonumber\\
    =&\Tr\left[(-i) \left(H_\zeta \lambda_{t,\zeta} -  \lambda_{t,\zeta} H_{-\zeta}\right)+G_{\zeta,-\zeta}\lambda_{t,\zeta}\right]\nonumber\\
    =&\Tr\left[(-i)\left( H_\zeta \lambda_{t,\zeta} - \lambda_{t,\zeta} H_{-\zeta}\right)+2\tilde G_{\zeta}\lambda_{t,\zeta}\right]\nonumber\\
    =& \langle \sqrt{\lambda_{t,\zeta}}, \left(\frac{1}{i}(H_\zeta-H_{-\zeta})+ 2\tilde G_\zeta\right)\sqrt{\lambda_{t,\zeta}}\rangle_{HS},\label{ptTr}
\end{align}
for every $\lambda\in \mathcal S_1^+ \cap \mathcal D(L_{\zeta,-\zeta})$. Since $H_{-\zeta}=H_\zeta^*$, for $\zeta=i\eta, \eta\in \mathbb R^n$ with $|\eta|=\nu, \nu \in (0,a)$, this yields 
\begin{equation}
\begin{aligned}
\partial_t\|\lambda_{t,\zeta}\|_1=&2\langle \sqrt{\lambda_{t,\zeta}}, (\im H_\zeta+\tilde G_\zeta)\sqrt{\lambda_{t,\zeta}} \rangle_{HS}\\
    \leq & 2\nu c'(\nu) \|\lambda_{t,\zeta}\|_{1}.
\end{aligned}
\end{equation}
Solving this inequality, with $\|\lambda_{t,\zeta}\|_1\vert_{t=0}=\|\lambda\|_1$, yields that
\begin{equation}\label{eq: betatzeta}
     \|\beta_{t,\zeta}(\lambda) \|_{1}\leq  e^{2\nu c't}\|\lambda\|_{1},
\end{equation}
$\forall \lambda\in \mathcal S_1^+\cap\mathcal D(L_{\zeta,-\zeta}),$ and consequently, by the B.L.T Theorem (\cite{RS1}, pp 9), we arrive at~\eqref{eq: betatzeta} for every $\lambda\in \mathcal S_1^+$. Now, using~\eqref{decom: lambda}, 
\begin{equation}
    \|\beta_{t,\zeta}(\lambda)\|_1\leq \sum\limits_{j=+,-}\left(\|\beta_{t,\zeta}(\lambda_j)\|_1+ \|\beta_{t,\zeta}(\lambda_j')\|_1\right)
\end{equation}
and~\eqref{eq: betatzeta} yields ~\eqref{beta acbd}.\end{proof}

Next, we estimate the first two factors on the r.h.s. of~\eqref{eq: 7.8}. Let $S^{n-1}$ denote the unit sphere in $\mathbb R^n$. Our starting point is relation~\eqref{eq: XbtY0} with $\zeta=- \theta=i\nu b $, where $b\in S^{n-1}$ to be chosen later on.
\begin{lemma} Let $U$ and $V$ be two bounded sets in $\Lambda$. For all $\zeta=i\nu b, \nu\in (0,a)$ and $b\in S^{n-1}$, 
\begin{equation}
    \|\chi_U T_{-\zeta}\|\|\chi_V T_{\zeta}\|\leq e^{-\nu \delta_{UV}},\label{eq: 7.9}
\end{equation}
where the constant $\delta_{UV}$ is given by
\begin{equation}
    \delta_{UV}:=r_U-\tilde r_V.\label{eq: 2.100}
\end{equation}
Here $r_U:=\inf\limits_{x\in U} b\cdot x$ and $\tilde r_V:=\sup\limits_{y\in V} b\cdot y$.
\end{lemma}
\begin{proof} By the definitions of $\chi_U$ and $T_\zeta$, we have
\begin{equation}
\|\chi_U T_{-\zeta}\|\leq \sup\limits_{x\in \Lambda} |\chi_U(x)e^{-i(-\zeta)\cdot x}|=\sup\limits_{x\in \Lambda} \left(\chi_U(x)e^{-\nu b\cdot x}\right).\label{6.11}
\end{equation}
Now, by the definition of $r_U$, ~\eqref{6.11} yields 
\begin{equation}
    \|\chi_U T_{-\zeta}\|\leq e^{-\nu r_U}.\label{52}
\end{equation}
Similarly, by the definition of $\tilde r_V$ in place of $r_U$, we obtain
\begin{equation}
    \|\chi_V T_\zeta\|\leq \sup\limits_{x\in V}e^{\nu b\cdot x}\leq e^{\nu \tilde r_V}.\label{53}
\end{equation}
Estimates~\eqref{52} and~\eqref{53} yield~\eqref{eq: 7.9}.\end{proof}

Estimates~\eqref{eq: 7.8} and~\eqref{eq: 7.9} and Proposition~\ref{prop: beta acbd} imply 
\begin{equation}
    \|\hat \chi_U \beta_t \hat \chi_V\|_1^{op}\leq 4e^{-2\nu\delta_{UV}+2\nu c't}.\label{est: S1opUbY}
\end{equation}
\begin{prop}Let $U\subset B_r(x_0)$ and $V\subset B_r(y_0)$ for some $x_0\in X$ and $y_0\in Y$, with $r=\frac{\epsilon}{2}d_{XY}, \epsilon\in (0,1)$. Then we have 
    \begin{equation}
        \|\hat \chi_U \beta_t \hat \chi_V\|_1^{op}\leq 4e^{-2\nu((1-\epsilon/2)d_{UV}-\epsilon d_{XY}-c't)}.\label{UYest}
    \end{equation}
\end{prop}
\begin{proof}

We translate both balls by the vector $y_0$ in order to place $y_0$ at the origin. Then we take $b=(x_0-y_0)/|x_0-y_0|$ and this gives 
\begin{equation}
    \begin{aligned}
        \delta_{UV}\geq &\inf\limits_{x\in B_{r}(x_0-y_0)} b\cdot x-\sup\limits_{y\in B_r(0)} b\cdot y \\
        \geq & |x_0-y_0| -\epsilon d_{XY}\\
        \geq &(1-\frac{\epsilon}{2}) d_{UV}-\epsilon d_{XY}.
    \end{aligned}\label{deltaUYbd}
\end{equation}
Estimates~\eqref{est: S1opUbY} and~\eqref{deltaUYbd} yield~\eqref{UYest}.
\end{proof}

For general sets $X$ and $Y$, we appeal to the following proposition proven in Subsection~\ref{sec: pf prop2.9}:
\begin{prop}\label{prop: 2.9}
Let $X$ and $Y$ be two arbitrary subsets of $\Lambda$, and assume~\eqref{UYest} holds for $U\subset B_r(x_0)$ and $V\subset B_r(y_0)$, for $r=\frac{\epsilon}{2}d_{XY}$, $\epsilon\in (0, \frac{2}{5})$, and for any $x_0\in X$ and $y_0\in Y$. Then~\eqref{eq: 4.7} holds, with $\mu=(1-\frac{5}{2}\epsilon)\nu$,
\begin{equation}
    c(\mu):=c'(\mu/(1-\frac{5}{2}\epsilon)), \label{def: cmu}
\end{equation}
where $c'(\nu)$ is defined in~\eqref{def: velocity bd}, and $c=c'(\mu/(1-\frac{5}{2}\epsilon))/(1-\frac{5}{2}\epsilon)=c(\mu)/(1-\frac{5}{2}\epsilon)$.
\end{prop}

\textnormal{Inequality}~\eqref{UYest} \textnormal{and} \textnormal{Proposition}~\ref{prop: 2.9} \textnormal{imply}~\eqref{eq: 4.7}, \textnormal{which completes the proof of} \textnormal{Theorem}~\ref{thm1}. \hfill $\qed$

\section{Inequalities for {\red normal and completely positive  
maps}}\label{sec: 3} 

We introduce generalizations of the completely positive and Lindblad maps $\psi'$  and $G'$ defined in {\red~\eqref{psi'} and} \eqref{G'} 
  Let $U=\{U_j \,:\, j=1,\cdots\}$ and $V=\{V_j \,:\, j=1,\cdots\}$ be collections of bounded operators on $\mathcal H$ s.t. $\sum\limits_{j=1}^\infty U_j^*U_j$ and $\sum\limits_{j=1}^\infty V_j^*V_j$ converge weakly. Define
the map \begin{equation}\label{psi'UV}   \psi'_{UV}(A):=\sum\limits_{j=1}^\infty V_j^*AU_j 
    \end{equation}
 for every $A\in \mathcal B$,      and its unital normalization $G_{UV}'$ satisfying $G_{UV}'(\one)=0$:
 \begin{equation}\label{G'UV}
G_{UV}'(A):=\psi_{UV}'(A)-\frac{1}{2}\{ \psi_{UV}'(\mathbbm1), A\}.
\end{equation}
{\red Clearly, $\psi'_{UV}$ and $G'_{UV}$  extend maps \eqref{psi'} and \eqref{G'}.} For want of a better term, we call the maps $\psi'_{UV}$ {\red normal} \st{sub-completely positive} maps.
\begin{lemma}\label{lem: GUV}The operators $\psi'_{UU}(\mathbbm1)$ and $\psi'_{VV}(\mathbbm1)$ are bounded and
    \begin{equation}
        \|G_{UV}'(A)\|\leq 3 \|A\| \|\psi_{UU}'(\mathbbm1)\|^{\frac{1}{2}}\|\psi_{VV}'(\mathbbm1)\|^{\frac{1}{2}}.\label{est: psiUVA}
    \end{equation}
\end{lemma}
\begin{proof} The operators $\psi_{UU}'(\mathbbm1)$ and $\psi_{VV}'(\mathbbm1)$ are bounded, since  
    \begin{equation}
        \psi_{UU}'(\mathbbm1)=\sum\limits_{j=1}^\infty U_j^*U_j,\label{def: psiUU'1}
    \end{equation}
    and similarly for $\psi'_{VV}(\mathbbm1)$. By~\eqref{est: 2.38}, we have 
    \begin{equation}
        \|\psi_{UV}'(A)\|\leq \|\sum\limits_{j=1}^\infty V_j^*V_j\|^{\frac{1}{2}}\|\sum\limits_{j=1}^\infty U_j^*A^*AU_j\|^{\frac{1}{2}},
    \end{equation}
    which can be rewritten as  
    \begin{equation}
        \| \psi_{UV}'(A)\|\leq \| \psi_{VV}'(\mathbbm1)\|^{\frac{1}{2}}\|\psi_{UU}'(A^*A)\|^{\frac{1}{2}}.
    \end{equation}
    Clearly, $\psi'_{UU}$ is a positive map. Hence, we have, for any self-adjoint operator $B$, 
\begin{equation}
\| \psi'_{UU}(B)\|\leq \|B\|\|\psi'_{UU}(\mathbbm1)\|.\label{est: psiUUB}
\end{equation}
Indeed, if we let $C:=\|B\|\mathbbm 1-B\geq 0,$ then we have
\begin{equation}
0\leq \psi_{UU}'(C)=\|B\|\psi'_{UU}(\mathbbm1)-\psi'_{UU}(B), 
\end{equation}
giving~\eqref{est: psiUUB}. Hence 
\begin{equation}
\|\psi_{UV}'(A)\|\leq \|A\|\|\psi_{UU}'(\mathbbm1)\|^{\frac{1}{2}}\|\psi_{VV}'(\mathbbm1)\|^{\frac{1}{2}},\label{est: psiUV'A}
\end{equation}
which yields 
\begin{equation}
\|\psi'_{UV}(\mathbbm1)\|\leq \| \psi'_{UU}(\mathbbm1)\|^{\frac{1}{2}}\| \psi'_{VV}(\mathbbm1)\|^{\frac{1}{2}}.
\end{equation}
Therefore, we obtain 
\begin{equation}
\| \{\psi'_{UV}(\mathbbm1), A\}\|\leq 2 \|A\|\| \psi'_{UU}(\mathbbm1)\|^{\frac{1}{2}}\| \psi'_{VV}(\mathbbm1)\|^{\frac{1}{2}},
\end{equation}
which together with estimate~\eqref{est: psiUV'A}, definition~\eqref{G'UV}, yields~\eqref{est: psiUVA}. \end{proof}
For $U$ and $V$ as in Lemma~\ref{lem: GUV}, define 
\begin{equation}
        G_{UV}(\rho)=\sum\limits_{j=1}^\infty \left( U_j \rho V_j^*-\frac{1}{2}\{ V_j^*U_j, \rho\}\right).
\end{equation}
{\red This map is dual to $G'_{UV}$ in the sense that $\Tr(AG_{UV}(\rho))=\Tr(G'_{UV}(A)\rho)$. With the latter relation, Lemma~\ref{lem: GUV} implies 
\begin{corollary}\label{cor GUV}    We have the estimate 
    \begin{equation}
\| G_{UV}(\rho)\|_1\leq 3{\red \|\psi_{UU}'(\mathbbm1)\|^{\frac{1}{2}}\|\psi_{VV}'(\mathbbm1)\|^{\frac{1}{2}}}\| \rho\|_1. 
    \end{equation}
\end{corollary}
Below, we use the following property of completely positive maps.}
\begin{lemma}\label{lem: 6.3} Let $\beta$ be a linear, completely positive map on $\mathcal S_1$. Then for any bounded operators $A$, $B$, $T$ and $V$ and for all $\rho\in \mathcal S_1^+$, we have
    \begin{equation}
        |\Tr(A\beta(T\rho V)B)|\leq (\Tr(A\beta(T\rho T^*)A^*))^{\frac{1}{2}} (\Tr(B^* \beta(V^*\rho V)B))^{\frac{1}{2}}.\label{6.17a}
    \end{equation}
\end{lemma}
 \begin{proof} We use that by the unitary dilation theorem (see~\cite{A}, Theorem 6.7), there exists a Hilbert space $\mathcal K$, a density operator $R$ on $\mathcal K$ and a unitary operator $J$ on $\mathcal H\times \mathcal K$ s.t. 
 \begin{equation}
\beta(\rho)=\Tr^{\mathcal K} (J(\rho\otimes R)J^*),\label{6.13b}
 \end{equation} 
 where $\Tr^{\mathcal K}$ is the partial trace in $\mathcal K$ (see e.g.~\cite{Ke}). For brevity, in the rest of this proof, we omit the tensor or product sign $\otimes$ in $\rho \otimes R$ and $\psi_i\otimes \varphi_j$, and write $\rho R$ and $\psi_i\varphi_j$, respectively. Substituting~\eqref{6.13b} into the l.h.s. of~\eqref{6.17a} and writing out the trace explicitly, we find
 \begin{align}
      \Tr(A \beta(T\rho V)B) =& \Tr_{\mathcal H\otimes \mathcal K} (A J(T\rho V  R)J^*B)\nonumber\\
      =& \sum\limits_{i,j} \langle \psi_i \varphi_j, A J(T\rho V R)J^*B\psi_i\varphi_j\rangle,
 \end{align}
 where $\{\psi_i\}$ and $\{\varphi_j\}$ are orthogonal basis in $\mathcal H$ and $\mathcal K$, respectively. Using the Cauchy-Schwarz inequality twice yields 
 \begin{align}
    & |\Tr(A \beta_t(T\rho V)B))|\nonumber\\
    \leq &\sum\limits_{i,j} \| (\rho^{1/2}T^* R^{1/2})J^* A^* \psi_i\varphi_j\|\|(\rho^{1/2}V R^{1/2})J^*B \psi_i\varphi_j\|\nonumber\\
    \leq & \left( \sum\limits_{i,j}  \| (\rho^{1/2}T^* R^{1/2})J^* A^* \psi_i\varphi_j\|^2 \right)^{\frac{1}{2}} \left( \sum\limits_{i,j}  \|(\rho^{1/2}V R^{1/2})J^* B\psi_i\varphi_j\|^2 \right)^{\frac{1}{2}}\nonumber\\
    =& \left(\sum\limits_{i,j} \langle \psi_i\varphi_j, AJ(T\rho T^*R)J^*A^*\psi_i\varphi_j\rangle\right)^{\frac{1}{2}}\nonumber\\
&\times\left(\sum\limits_{i,j} \langle \psi_i\varphi_j, B^*J(V^*\rho VR)J^*B\psi_i\varphi_j\rangle\right)^{\frac{1}{2}}.
 \end{align}
Since $\{ \psi_i\varphi_j\equiv \psi_i\otimes \varphi_j\}$ is an orthogonal basis in $\mathcal H \otimes \mathcal K$, this gives
 \begin{align}
 & |\Tr(A \beta_t(T\rho V)B))| \nonumber\\ \leq & \left(\Tr_{\mathcal H\otimes \mathcal K}(AJ (T\rho T^* R)J^*A^* )\right)^{\frac{1}{2}}\left( \Tr_{\mathcal H\otimes \mathcal K}(B^*J (V^*\rho V R)J^* B)\right)^{\frac{1}{2}}\nonumber\\
    =&\left( \Tr_{\mathcal H}(A\Tr^{\mathcal K}(J (T\rho T^*  R)J^*)A^* )\right)^{\frac{1}{2}} \left(\Tr_{\mathcal H}(B^*\Tr^{\mathcal K}(J (V^*\rho V  R)J^*)B )\right)^{\frac{1}{2}}.
 \end{align}
 Using~\eqref{6.13b} in the reverse direction, this yields 
 \begin{equation}
 |\Tr(A\beta_t(T\rho V)B)|\leq (\Tr (A\beta_t(T\rho T^*)A^*) )^{\frac{1}{2}} (\Tr (B^*\beta_t(V^*\rho V)B))^{\frac{1}{2}},
 \end{equation}
 which implies~\eqref{6.17a}.\end{proof}

 \section{Proof of Propositions~\ref{prop 4.1} and~\ref{prop: 2.9}}\label{sec: pf prop}
 \subsection{Proof of Proposition~\ref{prop 4.1}: The operators $L_{\zeta,\tilde \zeta}$ and $\beta_{t,\zeta,\tilde \zeta}$}\label{sec: pf 2.1}
 (a) The fact that the operator-family $L_{0,\zeta,\tilde \zeta}$ is bounded from $\hat R(\mathcal S_1)$ to $\mathcal S_1$ for every $\zeta,\tilde \zeta\in S_a^n$ follows from Condition~(AH). Corollary~\ref{cor GUV}, Eqs.~\eqref{eq.Grho} and~\eqref{4.6} and Condition (AW) on the $W_{j,\zeta}$'s imply
     \begin{equation}\label{est: Gzeta}
         \|G_{\zeta,\tilde \zeta}\|_{1}^{op}<\infty, \quad \forall \zeta,\tilde \zeta\in S_a^n,
      \end{equation}
     and therefore statement (a).  
     (b) By the definition of the operators $T_{\xi,\eta}, \xi, \eta\in \mathbb R^n$, we have, for any $\lambda\in \hat R(\mathcal S_1)=\{\hat R(\mu)\,: \, \mu\in \mathcal S_1\},$
     \begin{align}
     L_{\xi,\eta}\lambda=&-iT_\xi[H,T_\xi^{-1}\lambda T_\eta]T_\eta^{-1}\nonumber\\
     &+\sum\limits_{j=1}^\infty T_\xi\left( W_{j} T_\xi^{-1} \lambda T_\eta W_j^*-\frac{1}{2} \left\{ W_j^*W_j,T_\xi^{-1}\lambda T_\eta\right\} \right)T^{-1}_\eta\nonumber\\
     =& -i(H_\xi \lambda -\lambda H_\eta)\nonumber\\
     &+\sum\limits_{j=1}^\infty  \left( W_{j,\xi}\lambda W_{j,\eta}^*-\frac{1}{2}\left(W_{j,\xi}^*W_{j,\xi}\lambda+\lambda W_{j,\eta}^*W_{j,\eta}\right)\right),\label{eq.4.2}
     \end{align}
     where $W_{j,\xi}=T_\xi W_jT_\xi^{-1}$ and $W_{j,\xi}^*=T_\xi W_j^*T_\xi^{-1}=(W_{j,\xi})^*$.
     Eq.~\eqref{eq.4.2} and Conditions (AH) and (AW) imply that $L_{\xi,\eta}: \hat R(\mathcal S_1)\to \mathcal S_1$ has an analytic continuation in $\xi$ and $\eta$ from $\mathbb R^n \times \mathbb R^n$ to $S_a^n\times S_a^n$ and this continuation is of the form~\eqref{eq.form}-\eqref{eq.Grho}. \par

     (c) We recall the definition~\eqref{eq: L0zeta} of $L_{0,\zeta,\tilde \zeta}$ and
     \begin{equation}\label{eq: L0ztz}
     L_{0,\zeta,\tilde\zeta}=L^{\re}_{0}+L^{\im}_{0},
     \end{equation}
     where, for $\lambda\in \mathcal S_1$, 
     \begin{equation}\label{eq: Lre}
     L^{\re}_{0}\lambda:=-i\left( \re H_\zeta \lambda-\lambda \re H_{\tilde \zeta}\right)\text{ and } L^{\im}_{0}\lambda:=\im H_\zeta \lambda-\lambda \im H_{\tilde \zeta},
     \end{equation}
     with 
     \begin{equation}
     \re H_\zeta=\frac{1}{2}\left(H_\zeta+H_\zeta^*\right)\text{ and } \im H_\zeta=\frac{1}{2i}\left(H_\zeta-H_\zeta^* \right).
     \end{equation}
     Using~\eqref{eq: Lre}, we obtain 
     \begin{equation}\label{est: L1+G1}
     \|L_{0}^{\im}\lambda\|_1\leq \left(\|\im H_\zeta\|+\|\im H_{\tilde \zeta}\|\right)\|\lambda\|_1
     \end{equation}
     which, by Condition (AH) (see~\eqref{con: AH}), implies that, for all $\zeta,\tilde \zeta\in S_a^n, \quad \re\zeta=\re\tilde \zeta =0$,
     \begin{equation}
         \| L_{0,\zeta,\tilde \zeta}-L^{\re}_{0}\|_1^{op}\leq \|\im H_\zeta\|+\| \im H_{\tilde \zeta}\|<\infty.\label{re: 51}
     \end{equation}
Furthermore, 
     \begin{equation}
     e^{L^{\re}_{0}t}\lambda=e^{-i\re H_\zeta t} \lambda e^{i\re H_{\tilde \zeta}t}
     \end{equation}
     and therefore, since 
     \begin{equation}
     \begin{aligned}
       |e^{L^{\re}_{0}t}\lambda |=&|e^{-i\re H_\zeta t} \lambda e^{i\re H_{\tilde \zeta}t}|=|e^{-i\re H_{\tilde \zeta}t} \lambda e^{i\re H_{\tilde \zeta}t}|\\
       =&e^{-i\re H_{\tilde \zeta}t}|\lambda| e^{i\re H_{\tilde \zeta}t},
       \end{aligned}
     \end{equation} 
     we have
     \begin{equation}\label{eq: eLre}
     \| e^{L^{\re}_{0}t}\lambda\|_1=\|\lambda\|_1.
     \end{equation}
     Hence, by the standard Araki-Dyson perturbation expansion argument, $L_{\zeta,\tilde \zeta}$ generates the bounded evolution 
     \begin{equation}\label{def: 2.37a}
     \beta_{t,\zeta,\tilde \zeta}=e^{L_{\zeta,\tilde \zeta}t},\quad t\in \mathbb R 
     \end{equation}
     (for a precise bound, see Lemma~\ref{lem: +} below). (Another way to prove this is to use that $\sigma (L_{\zeta,\tilde \zeta}) \subset \{ z\in \mathbb C\, : \, |\re z|\leq C \},\, \forall (\zeta,\tilde \zeta)\in \mathcal S_a^n\times \mathcal S_a^n$, with $C=\sup\limits_{|\eta|=\nu, \, |\tilde \eta|=\nu }\| \im H_\zeta\|+\|\im H_{\tilde \zeta}\|+\|G_{\zeta,\tilde\zeta}\|_{1}^{op}$, and that for any $z \in \mathbb C$ with $|\re z|>C+1$, the following estimate
     \begin{equation}
         \|(L_{\zeta,\tilde \zeta}-z)^{-1}\|\leq (|\re z|-C-1)^{-1}\label{before box}
     \end{equation}
     holds, and then use the Hille-Yosida theorem.)

 The next lemma gives a precise bound on the one-parameter group $\beta_{t,\zeta,\tilde\zeta}$.
 \begin{lemma}\label{lem: +} For each $\zeta$, $\tilde \zeta\in S_a^n$. The operator $L_{\zeta,\tilde\zeta}$ generates the one-parameter group $\beta_{t,\zeta,\tilde \zeta}=e^{L_{\zeta,\tilde \zeta}t}$ and this group satisfies the estimate
 \begin{equation}
     \| \beta_{t,\zeta,\tilde \zeta}\|_{1}^{op}\leq e^{4t \left( \| G_{\zeta,\tilde\zeta}\|+\|\im H_\zeta\|+\|\im H_{\tilde \zeta}\|\right)}, \text{ for }\re\zeta=\re\tilde \zeta=0.\label{est: betazeta}
 \end{equation}
     
 \end{lemma}
 \begin{proof} Since $\beta_{t,\zeta,\tilde \zeta}=T_{\xi,\tilde \xi}\beta_{t,i\eta,i\tilde \eta}T_{\xi,\tilde \xi}^{-1}$, for $\zeta=\xi+i\eta$ and $\tilde \zeta=\tilde \xi+i\tilde \eta$, it suffices to consider $\zeta,\tilde \zeta\in S_a^n$ with $\re\zeta=\re\tilde \zeta=0$. Now, write out the Araki-Dyson-type series
 \begin{equation}
     \beta_{t,\zeta,\tilde \zeta}=e^{tL_0^{\re}}+\sum\limits_{j=1}^\infty I_{j,t},\label{eq: 30}
 \end{equation}
 where operators $I_{j,t}, j=1,\cdots,$ are given (with $t_0=t$), by
 \begin{equation}\label{eq: 31}
 \begin{aligned}
 I_{j,t}=&\int_0^{t}\int_0^{t_1}\cdots \int_{0}^{t_{j-1}} e^{(t-t_1)L_0^{\re}}(L_{\zeta,\tilde \zeta}-L_0^{\re})\cdots\\
 &\qquad \qquad \qquad \qquad\times e^{(t_{j-1}-t_j)L_0^{\re}}(L_{\zeta,\tilde \zeta}-L_0^{\re})e^{t_jL_0^{\re}}dt_j\cdots dt_1.
 \end{aligned}
 \end{equation}
 Taking the norm of this expression, using~\eqref{est: Lt4} and converting the integral over the simplex $\{ 0\leq t_j \leq t_{j-1}\leq \cdots \leq t_1\leq t\}$ to the integral over the $j$-cube $[0,t]^j$ yields
 \begin{equation}
  \begin{aligned}
 \|I_{j,t,\zeta,\tilde \zeta}\|_1^{op}\leq & \frac{4^jt^j}{j!}\left(\|L_{\zeta,\tilde \zeta}-L_0^{\re}\|_1^{op}\right)^j, \qquad t\geq 0,\, j=1,\cdots.
  \end{aligned}   \label{eq: 42}
 \end{equation}
 Eq.~\eqref{eq: 30}, together with the bound $\|e^{tL_0^{\re}}\|_{1}^{op}=1$ (see~\eqref{eq: eLre}), the relation $L_{\zeta,\tilde\zeta}-L_0^{\re}=L_0^{\im}+G_{\zeta,\tilde \zeta}$ (see~\eqref{eq.form} and~\eqref{eq: L0ztz}) and estimates~\eqref{est: Gzeta} and~\eqref{est: L1+G1}, implies~\eqref{est: betazeta}.\end{proof}

 Furthermore, using the Duhamel formula, we compute formally, for every $j$, 
     \begin{equation}
        \p_{\bar\zeta_j^{\#}} e^{L_{\zeta,\tilde \zeta}t}=\int_0^t e^{L_{\zeta,\tilde \zeta}(t-s)} \p_{\bar\zeta_j^\#}L_{\zeta,\tilde \zeta}e^{L_{\zeta,\tilde \zeta}s}ds=0,
     \end{equation}
     where $\zeta_j^\#=\zeta_j$ or $\tilde \zeta_j$. However, since, in general, the oeprators $L_{\zeta,\tilde \zeta}$ are unbounded, this formula has to be justified. We proceed differently.\par
     First, approximating $H_\zeta$ by bounded operators $H_\zeta(ia)(H+ia)^{-1}$, we can show that $e^{-iH_\zeta t}$ is analytic in $\zeta\in S_a^n, \forall t\in \mathbb R$. The latter implies that $e^{L_{0,\zeta,\tilde \zeta}t}\lambda =e^{-iH_\zeta t}\lambda e^{iH_\zeta t}$ is analytic in $\zeta,\tilde \zeta\in S_a^n$ for all $t\in \mathbb R$. Now, using the Duhamel principle
     \begin{equation}
         e^{L_{\zeta,\tilde \zeta}t}=e^{L_{0,\zeta,\tilde \zeta}t}+\int_0^t e^{L_{0,\zeta,\tilde \zeta}(t-s)} G_{\zeta,\tilde \zeta}e^{L_{\zeta,\tilde \zeta}s}ds
     \end{equation}  
     and analyticity of $e^{L_{0,\zeta,\tilde \zeta}t}$ and $G_{\zeta,\tilde \zeta}$ in $\zeta,\tilde\zeta\in S_a^n$, we find
     \begin{equation}
         \p_{\bar\zeta_j^{\#}} e^{L_{\zeta,\tilde \zeta}t}=\int_0^t e^{L_{0,\zeta,\tilde \zeta}(t-s)} G_{\zeta,\tilde \zeta}\p_{\bar\zeta_j^\#}e^{L_{\zeta,\tilde \zeta}s}ds,
     \end{equation}
     which implies that $\p_{\bar\zeta_j^\#}e^{L_{\zeta,\tilde \zeta}t}=0$ for $t$'s sufficiently small and therefore for all $t$'s.
     Hence, $e^{L_{\zeta,\tilde \zeta}t}$ is analytic as an operator-function of $(\zeta,\tilde \zeta)\in \mathcal S_a^n\times \mathcal S_a^n$.\par

 (d) Recall $\beta_{t,\zeta}\equiv \beta_{t,\zeta,-\zeta}$ and fix $\zeta=i\eta\in S_a^n$, $\eta\in \mathbb R^n$. By Lemma~\ref{lem: +}, we have $\beta_{t,\zeta}(\rho)\in \mathcal S_1$ for $\rho \in \mathcal S_1$. Next, we prove $\beta_{t,\zeta}(\rho)\geq 0$ for $\rho\in \mathcal S_1^+$. Let $\rho\in \mathcal S_1^+$ be s.t. 
 \begin{equation}\label{con: nu}
 \nu:=T_{-\zeta}\rho T_{-\zeta}\in \mathcal S_1^+. 
 \end{equation}
Since $T_{-\zeta}$ is self-adjoint and $\rho\geq 0$, we have that $\nu\geq 0$.

 Next, let $\psi\in \mathcal D(T_{\zeta})$. Then using the analytic continuation (in $\zeta$ and $\tilde \zeta$), we obtain that
 \begin{equation}
     \langle \psi, \beta_{t,\zeta}(\rho)\psi\rangle=\langle T_\zeta \psi, \beta_t(\nu) T_\zeta\psi\rangle\geq 0.
 \end{equation}
Since the set of all $\rho \in \mathcal{S}_1^+$ satisfying~\eqref{con: nu} is dense in $\mathcal{S}_1^+$, and $\beta_{t,\zeta}$ is bounded on $\mathcal{S}_1^+$ for all $\zeta \in S_a^n$, it follows that
\[
\beta_{t,\zeta}(\rho) \geq 0 \quad \text{for all } \zeta \in S_a^n \text{ with } \operatorname{Re} \zeta = 0.
\]
  \hfill $\qed$

 \subsection{Proof of Proposition~\ref{prop: 2.9}}\label{sec: pf prop2.9}

 Using the decomposition~\eqref{decom: lambda}, estimate~\eqref{eq: 4.7} can be reduced to proving 
 \begin{equation}\label{eq: 3.1}
 \|  \hat \chi_X \beta_t \hat \chi_Y \rho\|_{1}\leq Ce^{-2\mu(d_{XY}-ct)}\|\rho\|_{1},\quad \forall \rho \in \mathcal S_1^+,
 \end{equation}
 for some constant $C=C(n,c,\mu)>0$ depending on $n,c,\mu$. Let $\{{X_j}\}_{j=1}^{j=N_1}$ and $\{{Y_j}\}_{j=1}^{j=N_2}$ be decompositions of $X$ and $Y$, with ${X_j}, j=1,\cdots,N_1$, and ${Y_j}, j=1,\cdots,N_2$, containing in the balls centered at $x_j\in X$ and $y_j\in Y$, respectively, of the radius $r=\frac{\epsilon d_{XY}}{2}$. We have $\sum\limits_{k=1}^{N_1}\chi_{X_k}=\chi_X$ and $\sum\limits_{j=1}^{N_2}\chi_{Y_j}=\chi_Y$. Inserting the above partitions of unity into $\Tr(\hat \chi_X\beta_t\hat\chi_Y\rho),$ we find 
 \begin{equation}
 \Tr(\hat\chi_X \beta_t\hat \chi_Y \rho)=\sum\limits_{k=1}^{N_1}\sum\limits_{j_1=1}^{N_2}\sum\limits_{j_2=1}^{N_2} \Tr(\hat\chi_{X_k}\beta_t(\chi_{Y_{j_1}}\rho \chi_{Y_{j_2}})). 
 \end{equation}
 Using this and Lemma~\ref{lem: 6.3}, we obtain 
 \begin{align}
     0 \leq & \Tr(\hat \chi_X \beta_t(\hat \chi_Y \rho))\nonumber\\
     \leq&\sum\limits_{k=1}^{N_1} \sum\limits_{j_1=1}^{N_2}\sum\limits_{j_2=1}^{N_2}|\Tr(\hat \chi_{X_k} \beta_t( \chi_{Y_{j_1}} \rho\chi_{Y_{j_2}}))|\nonumber\\
     \leq & \sum\limits_{k=1}^{N_1}\sum\limits_{j_1=1}^{N_2}\sum\limits_{j_2=1}^{N_2}\left( \Tr(\hat \chi_{X_k} \beta_t(\hat\chi_{Y_{j_1}}\rho))\right)^{\frac{1}{2}}\left( \Tr(\hat \chi_{X_k} \beta_t(\hat\chi_{Y_{j_2}}\rho))\right)^{\frac{1}{2}}\nonumber\\
     =& \sum\limits_{k=1}^{N_1}\left( \sum\limits_{j=1}^{N_2} (\Tr(\hat \chi_{X_k} \beta_t(\hat \chi_{Y_{j}}\rho)))^{1/2}\right)^2.
     \end{align}
 By estimate~\eqref{UYest}, this yields 
 \begin{align}
     \Tr(\hat \chi_X \beta_t(\hat \chi_Y \rho))\leq &\sum\limits_{k=1}^{N_1} \left( \sum\limits_{j=1}^{N_2} 2e^{-\nu(1-\epsilon/2)d_{X_{k}Y_j}-\epsilon d_{XY}-\nu c't)}\| \hat\chi_{Y_j}\rho\|_{1}^{1/2}\right)^2\nonumber\\
     =:& 4e^{2\nu c't+2\nu \epsilon d_{XY}} M(\rho),\label{3.33b}
 \end{align}
 for any $\epsilon\in (0,1)$. To estimate $M(\rho)$, we proceed as in~\cite{SigWu}, Eqs (2.22)-(2.33). Namely, we let $\rho_j=\hat \chi_{Y_j}\rho$ and $\nu'=\nu(1-\epsilon/2)$ and write 
 \begin{align}
 M(\rho)=& \sum\limits_{k=1}^{N_1} \left( \sum\limits_{j=1}^{N_2}e^{-\nu'd_{X_{k}Y_j}}\|\rho_j\|_{1}^{\frac{1}{2}}\right)^2\nonumber\\
 =&  \sum\limits_{k=1}^{N_1} \sum\limits_{j_1=1}^{N_2} \sum\limits_{j_2=1}^{N_2}e^{-\nu'(d_{X_{k}Y_{j_1}}+d_{X_kY_{j_2}})}\|\rho_{j_1}\|_{1}^{\frac{1}{2}}\|\rho_{j_2}\|_{1}^{\frac{1}{2}}\nonumber\\
 \leq &   \sum\limits_{j=1}^{N_2} \|\rho_{j}\|_{1}C_{XY}
 \end{align}
 where $C_{XY}$ is given by (see~\cite{SigWu}, Eqs. (2.23)-(2.26))
 \begin{equation}
 C_{XY}:=  \sum\limits_{k=1}^{N_1} \sum\limits_{j_2=1}^{N_2} e^{-\nu'(d_{X_{k}Y_{j_1}}+d_{X_kY_{j_2}})}.
 \end{equation}
 By Eq.~(2.32) of~\cite{SigWu} and the relation $\sum\limits_{j=1}^{N_2}\|\rho_j\|_{1}=\sum\limits_{j=1}^{N_2}\Tr(\hat\chi_{Y_j}\rho)=\Tr \rho=1$, we have 
 \begin{equation}
 M(\rho)\leq C d_{XY}^{2(n-1)}e^{-2\nu' d_{XY}},
 \end{equation}
 for some constant $C=C(\nu', n, \epsilon)=C(\nu,n,\epsilon)>0$. This, together with~\eqref{3.33b} and the notation $\mu=\nu(1-\frac{5\epsilon}{2})$ and $c=\frac{c'}{1-\frac{5\epsilon}{2}}$, gives~\eqref{eq: 3.1} yielding therefore Proposition~\ref{prop: 2.9}. \hfill $\qed$


\bigskip

\paragraph{\bf Acknowledgment}

The first author is grateful to J\'er\'emy Faupin, Marius Lemm, Donghao (David) Ouyang and Jingxuan Zhang for discussions and enjoyable collaboration. The second author sincerely appreciates Avy Soffer's interest in this project. {\red The authors are grateful to the anonymous referee for important remarks and suggestions. In particular, Proposition B.1 was worked in response to one of the referee's comments.} 
The research of I.M.S. is supported in part by NSERC Grant NA7901.
X. W. is supported by 2022 \emph{Australian Laureate Fellowships} grant FL220100072. Her research was also supported in part by NSERC Grant NA7901 and NSF-DMS-220931.  
Parts of this work were done while the second author visited University of Toronto and Rutgers University.

\medskip

\paragraph{\bf Declarations}
\begin{itemize}
	\item Conflict of interest: The Authors have no conflicts of interest to declare that are relevant to the content of this article.
	\item Data availability: Data sharing is not applicable to this article as no datasets were generated or analysed during the current study.
\end{itemize}

\bigskip

\appendix 

\section{On existence theory for vNLE}\label{sec:vNLEexist}
The existence theory for the vNL equation is based on the decomposition 
    \begin{equation}\label{L0+G}
        L=L_0+G
    \end{equation}
    of the vNL generator $L$, with the von Neumann and Lindblad parts, $L_0$ and $G$, given by 
    \begin{equation}\label{L0G}
        L_0\rho=-i[H,\rho]\,\, \text{ and }\,\,G\rho=\sum\limits_{j=1}^\infty\left(W_j\rho W_j^*-\frac{1}{2}\left\{ W_j^*W_j,\rho\right\} \right).
    \end{equation}
    
    Under the conditions above the operator $L_0$ is defined on the set 
    \begin{equation}
     \mathcal{D}(L_0) := \left\{ \rho \in \mathcal{S}_1 \,:\,\begin{aligned}
    & \rho: \mathcal{D}(H) \to \mathcal{D}(H) \text{ and } [H, \rho] \text{ extends } \\
    & \text{from } \mathcal{D}(H) \text{ to } \mathcal{H} \text{ as an element of } \mathcal{S}_1
    \end{aligned}
    \right\}.
    \end{equation}
    The latter set contains the subset $\{ (H+i)^{-1}\rho (H-i)^{-1}\, :\, \rho \in \mathcal S_1\}$, which is dense in $\mathcal S_1$ (in the $\mathcal S_1$-norm $\|\lambda\|_1$.). For the second term, $G$, we observe that $\sum\limits_{j=1}^\infty W_j^*W_j$ is a bounded operator, as a weak limit of bounded operators (\cite{RS1}, Theorem~VI.1), and, for any $\rho\in \mathcal S_1^+$, $S_N:=\sum\limits_{j=1}^N W_j \rho W_j^*$ is an increasing sequence of positive, trace-class operators s.t. 
    \begin{equation}
       \|S_N-S_M\|_1= \Tr (S_N-S_M)=\Tr(\rho \sum\limits_{j=M}^N W_j^*W_j)\to 0,\label{Eq. WjWj}
    \end{equation}
    for $N>M\to \infty$, and therefore $S_N$ converges in the $\mathcal S_1$ norm as $N\to \infty$ and its limit $\sum\limits_{j=1}^\infty W_j\rho W_j^*$ is positive trace class operator. This way one can prove that $G$ is a bounded operator on $\mathcal S_1$ (see \cites{Davies5, OS} for details). Hence, the operator $L$ is well defined on $\mathcal D(L)=\mathcal D(L_0)$.  \par

    By the explicit representation $e^{L_0t}\rho=e^{-iHt}\rho e^{iHt}$, the operator $L_0$ generates a one-parameter group $\alpha_t=e^{tL_0}$ of isometries on $\mathcal S_1$ (von Neumann evolution), and therefore, by a standard perturbation theory, since $G$ is bounded, $L$ generates a one-parameter group of bounded operators, $\beta_t=e^{tL}$.

    In conclusion, we prove a bound on $e^{Lt}$ used in Section~\ref{sec: 2}. 
\begin{lemma} We have
    \begin{equation}
        \|e^{Lt}\|_{1}^{op}\leq 4.\label{est: Lt4}
    \end{equation}
\end{lemma}
\begin{proof} We use that every $\lambda\in \mathcal S_1$ can be decomposed as in~\eqref{decom: lambda}, with $\lambda_\pm$ and $\lambda_\pm'$ satisfying
\begin{equation}
\|\lambda_\pm\|_{1}\leq \|\lambda\|_{1}\quad \text{ and }\quad \|\lambda_\pm'\|_{1}\leq \|\lambda\|_{1}.\label{est: rhopm}
\end{equation}
Hence, it suffices to consider $\lambda\in \mathcal S_1$, s.t. $\lambda\geq 0$. Since $e^{tL}$ is a positivity and trace preserving map (see of \cite{OS}), we have 
\begin{equation}
  e^{tL}\lambda \geq 0\qquad \text{ and }\qquad  \Tr(e^{tL}\lambda)=\Tr \lambda.
\end{equation}
Hence, due to Eqs.~\eqref{decom: lambda} and~\eqref{est: rhopm}, we have
\begin{equation}
\begin{aligned}
  \| e^{tL}\lambda\|_{1}\leq &\| e^{tL}\lambda_+\|_{1}+\| e^{tL}\lambda_-\|_{1}+\| e^{tL}\lambda_+'\|_{1}+\| e^{tL}\lambda_-'\|_{1}\\
  =& \| \lambda_+\|_{1}+\| \lambda_-\|_{1}+\| \lambda_+'\|_{1}+\| \lambda_-'\|_{1}\\
  \leq & 4\|\lambda\|_{1}.
\end{aligned}
\end{equation}\end{proof}


{\red
\section{Positivity of $c(\mu)$}\label{sec: positivity}
Let $L'$ be the Heisenberg-Lindblad generator~\eqref{def: L'}, which is dual to the von Neumann-Lindblad one,~\eqref{L}; explicitly,
\begin{align}\label{L'}
L'(A)&=i[H,A]+G'(A),\\ 
\label{G'2} G'(A)&=\sum \left( W_j^*AW_j-\frac{1}{2}\{ W_j^* W_j,\, A\}\right).
\end{align}

Define the operator $L'(x)$ by $L'(x)=\frac{1}{i}\n_\xi L'(e^{i\xi\cdot x})\vert_{\xi=0}$. By Eq.~\eqref{eq: ptAt}, $v=L'(x)$ is the Heisenberg-Lindblad velocity  (\cites{Breteaux_2022,Breteaux_2023}). By Lemma~\ref{lem: cnu0} below, it is well defined.

\begin{prop}[Positivity of $c(\mu)$]\label{prop: pos cnu} $c(\mu)\geq 0$, $\forall\, \mu\in (0,a)$. Moreover, if $L'(x) \neq 0$, then  $c(\mu)>0$ for all $\mu\in (0,a)$.
\end{prop}
\begin{proof} By~\eqref{eq: tGzeta sys}, 
\begin{equation}\label{Kineq}\Im H_\zeta+\tilde G_\zeta\geq K_{\zeta},\text{ with }K_\zeta=\Im H_\zeta +G_\zeta^{\#},
\end{equation} 
where 
\begin{equation}
G_\zeta^{\#}:= \frac{1}{4} \sum_j (W_{j,\zeta}^* W_{j,\zeta} - W_{j,-\zeta}^* W_{j,-\zeta}).
\end{equation}
By the definition of $H_{i\eta}$, we have $\langle \psi, H_{i\eta}^*\psi\rangle=\langle \psi, H_{-i\eta}\psi\rangle$, for a dense set of $\psi$'s,  which implies $\Im H_{-i\eta}=-\Im H_{i\eta}$. Moreover, $G_{-i\eta}^\#=-G_{i\eta}^\#$ and therefore $K_{-i\eta}=-K_{i\eta}$. Hence, \begin{align}\label{Kalt}
\text{either }\ \langle\psi_0, K_{i\eta_0}\psi_0 \rangle>0\ \text{ or }\ \langle \psi_0, K_{-i\eta_0}\psi_0\rangle > 0,
\end{align} for some $\psi_0\in \mathcal H$ and $\eta_0\in \mathbb R^n,$ $|\eta_0|= \nu$, or $\langle \psi, K_{i\eta}\psi\rangle=0$ for all $\psi\in \mathcal H$, and  all $\eta\in \mathbb R^n, |\eta|= \nu$. Due to \eqref{Kineq}, this implies that either \\
\begin{equation}
c'(\nu)= \sup_{\zeta=i\eta, |\eta|=\nu} \sup\limits_{\psi\in \mathcal H}\langle \psi, (\Im H_\zeta+\tilde G_\zeta)\psi \rangle>0
\end{equation}
or 
\begin{equation}
\Im H_{i\eta}+\tilde G_{i\eta}=0\qquad \text{ for every }\eta\in \mathbb R^n, |\eta|=\nu.
\end{equation}
In the first case, we have $c'(\nu)>0$. The second case is ruled out by our condition $L'(x)\neq 0$. Indeed, by Lemma~\ref{lem: cnu0} below, $\nabla_\eta(\Im H_{i\eta}+\tilde G_{i\eta})\vert_{\eta=0}=L'(x)\neq 0.$

Finally, we observe that, by~\eqref{def: cmu}, $c'(\nu)\geq0$ (resp. $c(\nu)>0$) $\forall \nu \in (0,a)$ implies $c(\mu)\geq0$ (resp. $c(\mu)>0$) for all $\mu\in (0,a)$.\end{proof}

}


{\red
In the next lemma, we show that  the Heisenberg velocity operator $v=L'(x)$ is well defined and connect it to $c(\mu)$. 
\begin{lemma}\label{lem: cnu0} 
(a) $L'(x)=-i\n_{\xi}L'(e^{i\xi\cdot x})\vert_{\xi=0}$ is well defined and is given explicitly by 
\begin{equation}\label{L'x}
L'(x)=i[H,x]+\tilde G',
\end{equation}
where  
\begin{equation}\label{eq: def G'}
\tilde G':=\frac{1}{2}\sum\limits_{j=1}^\infty i(W_j^*W_j'-(W_j')^*W_j), 
\end{equation}
with $W_j':=i\nabla_\eta W_{j,i\eta}\vert_{\eta=0}=-i[x, W_j]$.

(b) With the notation $\sup A=\sup\limits_{\|\psi\|=1}\langle\psi, A\psi \rangle,$ $c(0)$ is given as
\begin{equation}\label{c0}
c(0)=\sup\limits_{\hat \eta\in S^{n-1}} \sup \big(\hat\eta\cdot L'(x)\big).
\end{equation}
\end{lemma}
 \begin{proof} (a) It is convenient to use a different representation of the map $G'$: 
\begin{align}\label{G'x}
G'(A)&=\frac{1}{2}\sum\limits_{j} ([W_j^*, A] W_j-W_j^*[W_j, A]).
\end{align}
The latter expression, together with the relation $-i \nabla_\xi i [W_j, e^{i\xi\cdot x}]\vert_{\xi=0}=i[W_j, x]=:W_j'$, yields $-i \nabla_\xi G'(e^{i\xi\cdot x})\vert_{\xi=0}=\tilde G'$, which, together with $-i \nabla_\xi i [H, e^{i\xi\cdot x}]\vert_{\xi=0}=i[H, x]$, implies \eqref{L'x}.


(b) Using that $\nabla_\eta H_{i\eta}\vert_{\eta=0}=-[H,x]$, which implies 
\begin{equation}\label{eq: *}
\nabla_\eta\im H_{i\eta}\vert_{\eta=0}=i[H,x],
\end{equation}
we obtain 
\begin{equation}
\Im H_{i\eta}=i[H,x]\cdot \eta+o(|\eta|).
\end{equation}
Furthermore, we expand, using~\eqref{def: tG}, with $\zeta=i\eta$, $\eta\in \mathbb R^n$, with $|\eta|\ll1$,
\begin{equation}\label{eq: tGeta}
\tilde G_{i\eta}=\tilde G'\cdot \eta+o(|\eta|),
\end{equation}
where $\tilde G'$ is given in~\eqref{eq: def G'}. Since $\Im H_\zeta\vert_{\zeta=0}=0$ and $\tilde G_\zeta\vert_{\zeta=0}=0$, we have that 
\begin{align}\label{c'1}
c'(\nu)&=\sup\limits_{\hat\eta \in S^{n-1}}\sup (\hat\eta\cdot\nabla_\eta (\Im H_{i\eta}+\tilde G_{i\eta}))\vert_{|\eta|=0} +o(\nu)\notag\\
&=\sup\limits_{\hat\eta\in S^{n-1}}\sup(\hat\eta\cdot( i[H,x]+\tilde G'))+o(\nu).
\end{align}
 This relation and Eq.~\eqref{L'x} imply~\eqref{c0}.\par

\end{proof}}

\end{document}